\documentclass[copyright,creativecommons]{eptcs}
\providecommand{\event}{QAPL 2012} % Name of the event you are submitting to
\usepackage[T1]{fontenc}
\usepackage{times}
\usepackage{amsfonts}
\usepackage{listings}
\usepackage{wrapfig} % put a figure in a paragraph
\usepackage{amsmath}
\usepackage{amsthm}
\usepackage{amssymb}

\newtheorem{lemma}{Lemma}[section]

\newtheorem{theorem}{Theorem}[section]
\newtheorem{hypo}{Condition}[section]
\newtheorem{definition}{Definition}[section]

\newcommand{\R}{\mathbb{R}}

\newcommand{\foo}{\hbox{\sl foo}}
\newcommand{\fooa}{\hbox{\sl foo}_a}
\newcommand{\fooA}{\hbox{\sl foo}_A}
\newcommand{\foob}{\hbox{\sl foo}_b}
\newcommand{\fooB}{\hbox{\sl foo}_B}
\newcommand{\ListFoo}{\hbox{\sl listFoo}}
\newcommand{\pe}{\hbox{\sl pathestimate}}
\newcommand{\graph}{\hbox{\sl graph}}
\newcommand{\cnt}{\hbox{\sl count}}
\newcommand{\mrk}{\hbox{\sl mark}}
\newcommand{\Cone}{C1}
\newcommand{\Ctwo}{C2}
\newcommand{\Cthree}{C3}
\newcommand{\Cfour}{C4}
\newcommand{\Kn}{k_{N^*}}
\newcommand{\En}{\epsilon_{N^*}}
\newcommand{\AND}{\mathrel{\wedge}}

\lstdefinelanguage{pscode}
  {morekeywords={while,for,return},
  sensitive=false,
  morecomment=[l]{//},
  morecomment=[s]{/*}{*/},
  morestring=[b]",
}

\renewenvironment{proof}{\noindent {\bf Proof}\quad}{\qed}
\title{A non-local method for robustness analysis of\\ floating point programs
      \thanks{This work has been partially supported by the project ANR-09-BLAN-0345-02  CPP.}}
\author{Ivan Gazeau, Dale Miller, and Catuscia Palamidessi\\
INRIA and LIX, Ecole Polytechnique}
\def\titlerunning{A non-local method for robustness analysis of
  floating point programs}
\def\authorrunning{I. Gazeau, D. Miller, and C. Palamidessi}
\begin{document}
\maketitle

\begin{abstract}
Robustness is a standard correctness property which intuitively means
that if the input to the program changes less than a fixed small
amount then the output changes only slightly.
This notion is useful in the analysis of rounding error for floating
point programs because it helps to establish bounds on output errors
introduced by both measurement errors and by floating point
computation.
Compositional methods often do not work since key constructs---like
the conditional and the while-loop---are not robust.
We propose a method for proving the robustness of a while-loop.
This method is non-local in the sense that instead of breaking the
analysis down to single lines of code, it checks certain global
properties of its structure.
We show the applicability of our method on two standard algorithms:
the CORDIC computation of the cosine and Dijkstra's shortest path
algorithm.\\ 

\end{abstract}

\section{Introduction}

%Motivations:

%Point 1: floating point analysis (cite main references)

Programs using floating point arithmetic are often used for critical
applications and it is therefore fundamental to develop methods to
establish the correctness of such programs.  A central problem in
dealing with floating point programs is the propagation of errors due
to the digitization of analog quantities and the introduction of
floating point errors during computation.  As is well known, floating
point arithmetic on these representations is quite different from real
number arithmetic: for example, addition is neither commutative nor
associative \cite{goldberg91surveys}.

%Point 2: robustness of programs (cite main references)

The developers of floating point programs would like to think in terms
of real number semantics instead of the more ad hoc and complicated
semantics given by some specific definition of floating point
arithmetic, such as the IEEE standard 754 \cite{ieee08}.  A central
problem in trying to reason about floating point programs is that in
dealing with non-continuous operators such as the conditional and the
while-loop, floating point errors can result in what appears to be
erratic behavior.  The problem is that these constructs are in general
\emph{non-robust}: small variations in the data can cause large
variations in the results.

When the program contains non-robust operators, traditional
compositional methods do not work well. Decomposing the
correctness of a looping program using Hoare triples, for example,
usually requires either introducing abstractions (e.g., approximations)
which can then make conclusions too imprecise, or to undergo a very
complex and intricate proof.

In this paper, we will take a different approach: we shall describe
some programs where such erratic behavior is recognized and find a way
to reason and bound all of that behavior.  By moving away from the
reasoning using Hoare's style emphasis on local and compositional
analysis of a looping program, we are able to avoid reasoning about
individual erratic behaviors: instead, we will treat such behaviors as
an aggregate and try to bound the behavior of that aggregate.

To illustrate such a possibility in reasoning, consider Dijkstra's
minimal path algorithm \cite{dijkstra59}.  This greedy algorithm moves
from a source node to its neighbors, always picking the node with the
least accumulated path from the source.  If one makes small changes to
the distances labeling edges, then the least path distance will change
also by a small amount: that is, this algorithm is continuous.
However, the actual behavior of the loop and the marking of subsequent
nodes can vary greatly with small changes to edge lengths.  Our
approach to reasoning will allow us to view all of these apparently
erratic choices of intermediate paths as an aggregate on which we are
able to establish the robustness of the entire algorithm.

\paragraph{Plan of the paper}
In the next section we introduce the concept of robustness and we
relate it to the notions of continuity and $k$-Lipschitz.
Section~\ref{sec:control} contains our main contribution: a schema for
reasoning about robustness in programs 
and its correctness. We then show the applicability of our proposal in
two main examples: The CORDIC algorithm for computing cosine,
presented in Section~\ref{sec:cordi}, and Dijkstra's shortest-path
algorithm, presented in Section~\ref{sec:dijkstra}.  In
Section~\ref{sec:related} we discuss some related
work. Section~\ref{sec:conclude} concludes and discusses some future
lines of research.

\section{Robustness of floating-point programs}

Robustness is a standard concept from control theory
\cite{pettersson96dc,parsec09}.  In the case of programming languages,
there are two definitions of robustness that have been considered.
One definition used by Chaudhuri et al \cite{chaudhuri10popl}
considered robustness to be based on continuity.  Later Chaudhuri et
al \cite{chaudhuri11sigsoft} considered a stronger notion of
robustness, namely the $k$-Lipschitz property: that is, 
changes to the input to a program lead to only proportionally bounded
changes to the output.
Another approach was used by Majumdar et al in
\cite{majumdar09rtss,majumdar10memocode} where robustness is
formulated as ``if the input of the program changes by an amount less
than $\epsilon$, where $\epsilon$ is a {\em fixed} constant, then the
output changes only slightly."  In our paper, we propose a more
flexible and general notion of robustness that generalizes both of
these concepts.  We now motivate and explain our notion of robustness
in more detail.

The notions of robustness considered in
\cite{chaudhuri10popl,chaudhuri11sigsoft} are mainly useful for
\emph{exact semantics}, namely when we do not take into account the
errors introduced by the representation and/or the computation. In
this case, the only deviation  comes from the  error of the input. The
continuity property, that for a function $f$ on reals is defined as:  
\[\forall \epsilon>0\; \exists \delta\; \forall i,i'\in \R\;  \;
|i-i'|<\delta \Rightarrow |f(i)-f(i')| < \epsilon\]
ensures that the correct output can be approximated when we can
approximate the input closely enough.   This notion of robustness,
however, is too weak in many settings, because a small  variation in
the input can cause an unbounded change in the output. The
$k$-Lipschitz property, defined as
\[\forall i,i' \in \R\; \;|f(i)-f(i')| \leq k |i-i'|\]
amends this problem because it bounds the variation in the output
linearly by the variation in the input.

%These notions are useful even for non exact semantic while inner error can be neglectable.

%However, we are interested in a better control of the error and we want to reason directly in a model for non exact semantic. 

In our setting, however, the $k$-Lipschitz property is too
strong. This is due to the following  reasons:
\begin{enumerate}
%\item We want to reason about \emph{approximated semantics}, such as those representing the floating-point implementations. Here, the variations in the inputs  are produced by  representation  errors which by their nature (transforming a datum from a continuous topology into one in  a discrete topology) cannot be infinitesimal, and therefore make the $k$-Lipschitz property inappropriate for robustness. 
\item If we consider a \emph{finite precision semantics}, like
  floating point implementations, the constant factor $k$ can become
  much bigger than the one optimal for the exact semantics. For
  instance, assume that the available representations are the numbers
  in the set $\{ k\,2^{-32} | k \in \mathbb{Z}\}$ and rounding is done
  by taking the lower value, and observe that a function like $f~:~x
  \mapsto 2^{-4}x$, which is $2^{-4}$-Lipschitz in the exact semantics,
  is only $1$-Lipschitz in this approximate semantics. Indeed, there
  exist two values that differ by just $2^{-32}$ and return a result
  that differ by $2^{-32}$. For example, take $1$ and  $1- 2^{-32}$:
  we have that $f(1)=2^{-4}$ and $f(1- 2^{-32})=2^{-4}-2^{-36}$, but
  the   second result will be rounded down to $2^{-4}-2^{-32}$.
%Here, the variations in the results are produced by  rounding  errors which by their nature (transforming a datum from a continuous topology into a discrete topology nature contradict the 
% $k$-Lipschitz property.
 %  cannot be satisfied for any finite $k$. due to the inner errors, any interval of size $0$ will produce a non empty interval which means $k=\infty$. 

\item There are algorithms that have a desired precision $e$ as a
  parameter and are considered correct as long as the result differs
  by at most $e$ from the results of the mathematical function they
  are meant to implement.  A program of this kind may be discontinuous
  (and therefore not $k$-Lipschitz) even if it is considered to be a
  correct implementation of a $k$-Lipschitz function.  The phenomenon
  is illustrated by the following program $f$ which is meant to
  compute the inverse of a strictly increasing function
  $g:\R^+\rightarrow \R^+$ whose inverse is $k$-Lipschitz for some
  $k$.
{\begin{lstlisting}[language=pscode]
f(i){ y=0;
      while(g(y) < i){
	   y = y+e; }
      return y; }
\end{lstlisting}}   
 
The program $f$ approximates  $g^{-1}$ with precision $e$  in the sense that 
\[ 
\forall x\in \R^+ \; f(x) -e \leq g^{-1}(x) \leq f(x)
\]
Given the above inequality, we would like to consider the  program $f$
as robust, even though the function it computes is discontinuous (and
hence not $k$-Lipschitz, for any $k$). 

\end{enumerate}

These two observations lead us to define another property,
$P^1_{k,\epsilon}$, to capture robustness:  
\[\forall i,i' \in \R, |f(i)-f(i')| \leq k |i-i'| + \epsilon \]
This property amends the two previous problems by setting $\epsilon$ to
$2^{-32}$ in the first example and to $e$ in the second example.
It also extends the usual definition of the $k$-Lipschitz property, which can be expressed as $P^1_{k,0}$. 

Now, we want to extend this definition to allow for several variables
and for other metric spaces besides $\R$: e.g., probability
distributions, intervals arithmetic etc.  Thus, we consider, instead,
two metric spaces: one for input ($I$, $d_I$) and the other for the
return value ($R$,$d_R$).  Hence, our robustness property
$P^2_{k,\epsilon}$ becomes
\[\forall i,i' \in I, d_R(f(i),f(i')) \leq k d_I(i,i') + \epsilon \]

Finally, since we are studying small deviation, it is not useful to
get this property for any $i$ and $i'$ in $I$ but rather when they are
close: i.e., $d_I(i,i')\le \delta$, for suitable values $\delta\in\R^+$.
In convex spaces, this property can be easily extended to  
pairs of inputs  having distance more than $\delta$ by using intermediate
values.  So, finally, in this paper we propose the property $P_{k,\epsilon,\delta}$, 
described in the following definition.

\begin{definition}\label{def:k-epsilon-delta}
Let $I$ and $R$ be metric spaces with distance $d_I$ and $d_R$
respectively, $f : I \to R$ a function, $k,\epsilon\in\R^+$, and let
$\delta\in\R^+\cup\{+\infty\}$.  We define the
property $P_{k,\epsilon,\delta}$ for the function $f$ as follows: 
\[\forall i,i' \in I,\quad d_I(i,i') \leq \delta\implies 
                           d_R(f(i),f(i')) \leq k d_I(i,i') + \epsilon \]
\end{definition}

%Classical notions of robustness for an exact semantic are often
%uniform continuity and $k$-Lipschitz property. Uniform continuity
%grants, in an exact 

\section{A schema and its correctness}\label{sec:control}
The main characteristic  of our schema is to subdivide the code into
several parts instead of analyzing it line by line. 
Our template, which we show in a moment, divides the data structures
in an algorithm into two parts, called $A$ and $B$.  Here, $A$ is the
witness to the progress of the algorithm: in particular, the stopping
condition will only depend on $A$ (and the input).   The structure $B$
is used to accumulate results that provide the answer when the
stopping condition is satisfied.

\subsection{The schema structure definition}\label{sec:schema}
Instead of
presenting a formal definition of program schema and matching of
code, we illustrate these with the schema in Figure~\ref{fig:main}.  

\begin{wrapfigure}{l}{60mm}
\begin{lstlisting}[language=pscode]
foo(i){
     a = a0;
     b = b0;
     while(S(i,a)){
	  c = O(a,b,c,i);
	  a = M(a,c);
	  b = N(i,b,c);
     }
     return b; }
\end{lstlisting}
\caption{The main template}\label{fig:main}
\end{wrapfigure}

Here, the schema variables {\tt a}, {\tt b}, {\tt c}, etc, denote
tuples of program variables such that no program variable occurs
twice among these schema variables. Program expressions such as 
\begin{lstlisting}
c = O(a,b,c,i);
\end{lstlisting}
denotes a program phrase that computes new values for the variables denoted
by {\tt c} from values of variables in the tuples {\tt a}, {\tt b},
{\tt c}, and {\tt i}.  The actual computation here will be denoted by
{\tt O}.  This looping program initializes the variables in {\tt a} and
{\tt b} with the values in the tuples {\tt a0} and
{\tt b0}, respectively.  The stopping condition for the loop is given
by the boolean valued expression {\tt S(i,a)} and the result of the
program is the tuple of values denoted by the variables in {\tt b}. 

We shall assume that all program variables are typed in the usual way:
variables may range over the values in their associated type.  Our
analysis of the metric properties of a looping program will, however,
consider that tuples of variables, for example, {\tt a} and {\tt b} in
Figure~\ref{fig:main}, range over some {\em metric space} on the
Cartesian product of the variables in the tuple.

\subsection{A sufficient condition for robustness}\label{sec:robustness}

% We denote by $i$ the set of variable in $I$ for which we want our
% property to be proved. If the program have additionnal arguments they
% can be considered as constants of the program. 

\begin{wrapfigure}{r}{60mm}
\begin{lstlisting}[language=pscode]
ListFoo(i){
     a = a0;
     b = b0;
     j = 0;
     while(! S(i,a)){
	  c = O(a,b,c,i);
          j = j+1;
          l[j] = c;
	  a = M(a,c);
	  b = N(i,b,c); }
     return l; }
\end{lstlisting}
\caption{Collecting {\tt c} values in a list}\label{fig:listfoo}
\end{wrapfigure}

We shall now prove that a program having the generic structure of $\foo$
given in Figure~\ref{fig:main} has, under certain conditions, the
property $P_{k,\epsilon,\delta}$ for some $k,\epsilon,\delta$.

The aim of our method is to postpone the analysis of the exact
semantics of commands as far as possible. In order to begin the
analysis without specific knowledge of this semantics, we need to
manipulate other programs made from the functions $O$, $M$, and $N$
that have been identified.  For example, the program $\ListFoo$ in
Figure~\ref{fig:listfoo} will be used to extract the list of values
of $c$ obtained for a particular execution of $\foo$ with input $i$.
The new lines added to $\ListFoo$ will assume the usual semantics for
natural numbers.

%Let $P$ be a property of one argument where that argument has type
%$I\rightarrow B$ (eg, $k$-Lipschitz, continuous, etc) which is local:
%if $P$ holds for any neighbored of $x \in I$ then $P$ holds for any $x \in I$.

%This definition extends the definition of $k$-Lipschitz as $P_{k,0}$
%means exactly $k$-Lipschitz. But because programs have to deal with
%rounding, most programs have discontinuity at the precise point very
%a rounding is done in one side or the other. This discontinuity is
%harmless as it is never bigger than the wished precision. That is why
%we allow a constant $\epsilon$ variation whatever is the distance
%between the two inputs. 

%Our goal in this note is to show that programs such as the following 
%one, called $\foo$, satisfies the property $P_{k_0,\epsilon_0}$ for
%some $k,\epsilon_0$.  Here, the function computed by $\foo$ is the
%one for which given $i\in I$, we 
%execute $\foo$ with input $i$.  The value returned by $\foo$ when it
%terminates is the value of function on $i$.

We now define two new programs.  The first is the 
$\foob$ program given below: it has the same shape as $\foo$
but instead of setting {\tt c} by the computation of {\tt O(a,b,c,i)}, it
sets {\tt c} with the values of a list given in input. Naturally, the stop
condition for the loop is now that all elements of the list have been
accessed. Note that since $a$ was just used in the computation of $O$,
the commands affecting $a$ are now useless and can be removed. 
%They have been commented with the $//$ syntax.  
\begin{lstlisting}[language=pscode]
foo_b(l,i){
//   a = a0;
     b = b0;
     for(int j = 0; j < l.length; j++ ){
	  c = l[j];
//	  a = M(a,c);
	  b = N(i,b,c); }
     return b; }
\end{lstlisting}
We have used Java-style instructions such as $l.\mathit{length}$ for the
length of the list $l$ and $l[j]$ for the $j^{th}$ element of the list
$l$.  (The $//$ syntax is used to form a comment.)  We define the new
function $\fooB(i,i') = \foob(\ListFoo(i),i')$.  
Notice that $\fooB(i,i)=\foo(i)$. 

The second program $\fooa(l)$ is the same program as $\foob$ except
that $a$ is returned instead of $b$. In this program, the lines where
$b$ is set are now useless. 
\begin{lstlisting}[language=pscode]
foo_a(l){
     a = a0;
//   b = b0;
     for(int j = 0; j < l.length; j++ ){
	  c = l[j];
	  a = M(a,c);
//	  b = N(i,b,c);
     }
     return a; }
\end{lstlisting}
Finally, we define $\fooA(i) = \fooa(\ListFoo(i))$. 
The two function $\fooA$ and $\fooB$ and relations between them will
be used to indirectly analyze the program $\foo$.

%To write our properties in a more functional fashion we introduce a new notation.
%\begin{definition}
%The function $\FOLD$ is defined in the usual way (in the functional
%programming literature):  $\FOLD(F,[],x)= x$ and 
%$\FOLD(F, h:q, x)=\FOLD(F, q, F(h,x))$, 
%where $F$ is a function of two arguments, the colon denotes the
%non-empty list constructor.
%\end{definition}

In what follows, we use the following conventions: the domain of the
variables $a$, $b$, $c$, and $i$ are $A$, $B$, $C$ and $I$,
respectively, and $a0$ and $b0$ are some determined constants of type
$A$ and $B$ respectively.  For every type $X$, the expression $X^*$
denote the type of lists of type $X$.

We now introduce four conditions that need to hold to prove that the
$\foo$ program satisfies $P_{k,\epsilon,\delta}$ for appropriate
values of $k$, $\epsilon$, and $\delta$.  These conditions apply to
{\em eight} parameters: namely, $\delta,\Kn,\En,K_A,\epsilon_2$,
$K_s,\epsilon_s,\epsilon_t$.  Condition $\Cone$ expresses the property
$P_{\Kn,\En,\delta}$ for the transformed program $\fooB$, condition
$\Ctwo$ expresses the fact that there is a relationship between the
values stored in $A$ and the values stored in $B$, and condition
$\Cthree$ and $\Cfour$ address the stability of the stop condition
$S(i,a)$.

\begin{hypo}[\Cone]\label{C1} \quad
$\forall l \in C^*. P_{\Kn,\En,\delta}(\lambda z. \foob(l,z))$.
\end{hypo}

% Now, we need to prove that the variables in $a$ can control the
% branching procedure. 

The next condition states that whenever two inputs $i$ and $i'$ are
within a $\delta$ of each other then it is the case that if their
images in $A$ (under $\fooa$)  are close, then their images in $B$
(under $\foob$) are close.

\begin{hypo}[\Ctwo]\label{C2}
\[%\exists k_a, \epsilon_2, \delta \in \R,
\forall i_1,i \in I, d_I(i,i_1) \leq \delta \implies
d_B( \fooB(i,i) , \fooB(i_1,i)) \leq k_A  d_A( \fooA(i_1) , \fooA(i)) + \epsilon_2
\]
\end{hypo}
%The condition express there exists a mapping between $A$ and $B$.

The stopping condition $S$ should satisfy the following two
conditions.  The first expresses that the boundary of the region $\{a\;
|\; S(i,a)\}$ cannot vary too much.
\begin{hypo}[\Cthree]\label{C3}
\[%\exists k_s, \epsilon_s \in \R,
\forall a \in A, \forall i,i' \in I,
d_I(i,i') \leq \delta \AND 
S(i',a) \implies \exists a' \in A, d_A(a,a') \leq k_s d_I(i',i) + \epsilon_s \AND 
S(i,a') \]
\end{hypo}

The following condition on $S$ states that the diameter of the region
$\{a\; |\; S(i,a)\}$ is as small as the desired precision.
\begin{hypo}[\Cfour]\label{C4}
\[%\exists \epsilon_t \in \R,
\forall a,a' \in A,\forall i \in I, S(i,a) \AND S(i,a') \implies d_A(a,a') \leq \epsilon_t\]
\end{hypo}

Finally, our main theorem is the following.
\begin{theorem}
If the program $\foo$ terminates and the conditions \Cone, \Ctwo,
\Cthree, and \Cfour\ hold, then $P_{k_0,\epsilon_0,\delta}$ holds for
the function computed by $\foo$ with $k_0= \Kn + k_A k_s $ and
$\epsilon_0= \En + k_A (\epsilon_s + \epsilon_t) + \epsilon_2$.
\end{theorem}

\begin{proof}
In the proof, we will use these two observations:
\begin{enumerate}
\item \label{O2} Since $\ListFoo(i)$ is obtained from the computation of
  $\foo(i)$, and since $\fooB(i,i')$ replaces the result of $O$ by this
  list, if we compute $\fooB(i,i)$ we are replacing each value for $c$
  by itself. Therefore we have that $\foo(i) = \fooB(i,i)$. 

%\item \label{O3} $\fs{\foo_t(i0,i)} = \FOLD(M,\ListFoo(i0),a0)$.
%\item \label{O4} As the value of $a$ never depends on $i$ in $\foo_t$, we have  :$ \forall i0, i, i'\in I, \fs{\foo_t(i0,i)} = \fs{\foo_t(i0,i')}$.

\item \label{O5} In the execution of $\foo(i)$, the final value of $a$
  that satisfies the stopping condition $S(i,a)$ is $\fooA(i)$. 
\end{enumerate}

By the observation~\ref{O2}, proving the theorem is equivalent to
proving 
\[\forall i,i0 \in I, d_I(i,i0) \leq \delta \implies
d_B(\fooB(i,i),\fooB(i0,i0)) \leq k_0 d_I(i,i0) + \epsilon_0.\]
By condition \Cone, choosing $l=\ListFoo(i0)$, we have 
\[ \forall i,i0 \in I,d_I(i,i_0) \leq \delta \implies
d_B( \foob(\ListFoo(i0),i0), \foob(\ListFoo(i0),i))\leq \Kn d_I(i,i0) + \En.\]
By definition of $\fooB$, we have
\begin{equation}\label{D}
\forall i,i0 \in I,d_I(i,i_0) \leq \delta \implies
d_B(\fooB(i0,i0),\fooB(i0,i)) \leq \Kn d_I(i,i0) + \En.
\end{equation}
From observation~\ref{O5}, $S(i0,\fooA(i0))$ holds.
% By the third observation, we derive that $S(i0,\foo_a(i0,i))$ holds.
By condition~\Cthree \, (instantiating $i'$ with $i0$) we derive that:
\begin{equation}\label{A}
\forall i,i0 \in I,d_I(i,i_0) \leq \delta \implies
\exists a' \in A, d_A(\fooA(i0),a') \leq k_s d_I(i,i0) + \epsilon_s \AND S(i,a').
\end{equation}
%Observe that $S(i,\fs{\foo(i)})$ holds, by definition. 
Hence, by observations~\ref{O5} and ~\ref{O2}, $S(i,\fooA(i))$ also
holds. From inequality~\eqref{A} and condition~\Cfour, we derive
\begin{equation}
d_A(a',\fooA(i)) \leq \epsilon_t.
\end{equation}
From the last inequality and from  inequality~\eqref{A}, we derive, using
the triangle inequality
\begin{equation}\label{B}
 d_A(\fooA(i0),\fooA(i)) \leq  k_s d_I(i,i0) + \epsilon_s + \epsilon_t.
\end{equation}
%Now observe that, by observation~\ref{O3} we have :
%\[\fs{\foo_t(i0,i)} = \FOLD(M,\ListFoo(i0),a0) \]
%\[\fs{\foo_t(i,i)} = \FOLD(M,\ListFoo(i),a0) \]
%And by observation~\ref{O1}
%\[ \FOLD(N~i,\ListFoo(i0),b0) = \foo_t(i0,i)\]
%\[ \FOLD(N~i,\ListFoo(i),b0) = \foo_t(i,i)\]
From condition~{\Ctwo}  and inequality~\eqref{B}, we have
\begin{equation}\label{C}
\forall i,i0 \in I, d_I(i,i0) \leq \delta \implies
d_B(\fooB(i0,i),\fooB(i,i)) \leq k_A (k_s d_I(i,i0) + \epsilon_s + \epsilon_t) +  \epsilon_2.
\end{equation}
From inequalities~\eqref{D} and \eqref{C}, using the triangle inequality, we derive
\[
\begin{array}{c}\label{eq:E}
\forall i,i0 \in I, d_I(i,i0) \leq \delta \\  \implies   \\
d_B(\fooB(i,i),\fooB(i0,i0)) \leq  \Kn d_I(i,i0) + \En + k_A (k_s d_I(i,i0) + \epsilon_s + \epsilon_t)  + \epsilon_2.
\end{array}
\]
Finally, we define $\epsilon_0= \En + k_A (\epsilon_s + \epsilon_t) +
\epsilon_2$ and $k_0 = \Kn + k_A k_s $.
%\qed
\end{proof}

\section{Example:  the CORDIC algorithm for computing cosine}\label{sec:cordi}
%These example are aimed to show this method is
%general enough to deal with algorithms that are very different from
%each other.
In this section  we apply our method to a program implementing  the CORDIC algorithm~\cite{volder59}, and we prove that  it is
$P_{k,\epsilon,\infty}$.

CORDIC (COordinate Rotation DIgital Computer) is a class of simple and
efficient algorithms to compute  hyperbolic and trigonometric
functions using only basic arithmetic (addition, subtraction and
shifts), plus table lookup. The notions behind this computing
machinery were motivated by the need to calculate the trigonometric
functions and their inverses in real time navigation systems. Still
now-a-days, since the CORDIC algorithms  require  only  simple integer
math, CORDIC is the preferred implementation of math functions on
small  hand calculators.

CORDIC is a successive approximation algorithm:  A sequence of
successively smaller rotations based on binary decisions drives the
algorithm towards 
the value we want to find.  The CORDIC version illustrated in the
program below  computes the cosine of any angle in $[0,\pi/2]$.  

{%\scriptsize
\begin{lstlisting}[language=C]
double cos(double beta)
{
    double x = 1, y = 0, x_new, theta = 0, sigma, e = 1E-10; 
    int Pow2=1; 
    while(|theta - beta| > e) {
         Pow2 *= 2;
         if(beta > theta) 
	   sigma=1;
	 else
	   sigma=-1;
	 sigma=sigma/Pow2;
	 theta += atan(sigma); // Value stored
	 fact= cos(atan(sigma)); // Value stored 
         x_new = x + y*sigma;
         y = fact * (y - x*sigma);
         x = fact * x_new; }
    return x; }


 
\end{lstlisting}}
Note that this program makes call to trigonometric functions like
cosine itself. But in the actual implementation, as it is explained in
the comments, these calls (that are done on values divided by successive powers of two) are
stored in a database so that no computation of these functions is actually done.

\subsection{Scheme instantiation}
To apply our method, we have first of all to instantiate the schema variables  $A$, $B$, $C$  (cf. Section~\ref{sec:schema}) with a suitable partition of the variables of the program.
The variables $I$ are determined: they must be instantiated  with the variables which represent the  
input. 

In this example the partition for the variables will be the
following.

{%\scriptsize
\begin{lstlisting}
A := double theta;
B := double x,y;
C := double sigma;
I := double beta;
\end{lstlisting}}

We now must define a suitable metric on the types of the  variables in $A$ and $B$. We choose the  following:
\begin{itemize}
\item $d_A$ is the usual distance on $\R$.
\item $d_B$ is the $L_2$ norm on $\R^2$.
\end{itemize}

Now we need to identify the stopping condition $S(i,a)$. This is given by:
\begin{lstlisting}
S(beta,theta) := | theta - beta | <= e
\end{lstlisting}

Then, we need to instantiate the functions $M(a,c)$, $N(i,b,c)$, $O(a,b,c,i)$ of the schema with suitable regions of code.
We choose these as follows:

{%\scriptsize
\begin{lstlisting}[language=C]
O(theta,<x,y>,sigma,beta) {
    Pow2 *= 2;
    if(beta > theta) 
      sigma=1;
    else
      sigma=-1;
    sigma=sigma/Pow2;
    return sigma; }
    
M(theta,sigma) {
    theta += atan(sigma);
    return theta; }
    
N(beta,<x,y>,sigma) {
    fact = cos(atan(sigma));
    x_new = x + y*sigma;
    y = fact * (y - x*Pow2);
    x = fact * x_new;
    return <x,y>; }

 
\end{lstlisting}} 

Finally, we need to prove that the conditions \Cone, \Ctwo, \Cthree, and \Cfour \, (cf. Section~\ref{sec:robustness}) are satisfied. 

\subsection{Proofs of the conditions}

\paragraph{\Cone: $\forall l \in C^*. P_{\Kn,\En,\delta}(\lambda z. \foob(l,z))$}
This condition can be proved for the following program by such standard
techniques as abstract interpretation or Hoare triples.

{%\scriptsize
\begin{lstlisting}[language=C]
double cos(double beta, int[] listFoo)
{
    double x = 1, y = 0, x_new, theta = 0, sigma = 0,e = 1E-10; 
    int Pow2=1; 
    for(int j=O;j<listFoo.length;j++) {
	 sigma=listFoo[j];
	 fact = cos(sigma);
         x_new = x + y*sigma;
         y = fact * (y - x*sigma);
         x = fact * x_new;
    }
    return x*K;
}


 
\end{lstlisting}}

\paragraph{\Ctwo: 
$\forall i_1,i \in I, d_I(i,i_1) \leq \delta \implies
d_B( \fooB(i,i) , \fooB(i_1,i)) \leq k_A  d_A( \fooA(i_1) , \fooA(i)) + \epsilon_2$}
This part of the proof is rather technical. The interested reader can
find it in the appendix of \cite{gazeau12hal}. 
The proof of ~\Ctwo \, is the most difficult part of this example. We have proved it ``by hand'', and we do not
claim that there is an easy way to automate it.  
However, this proof points out that  we can prove
the intended property  without considering the whole semantics of the program,  but
just the relevant properties. 

\paragraph{\Cthree: $\forall a \in A, \forall i,i' \in I,
d_I(i,i') \leq \delta \AND 
S(i',a) \implies \exists a' \in A, d_A(a,a') \leq k_s d_I(i',i) + \epsilon_s \AND 
S(i,a')$}
The instantiation of $S(i,a)$ corresponds to $|i-a| \leq e$, so \Cthree\, is given by the condition:
\[\forall a \in A, \forall i,i' \in I, |i-a|\leq e, \exists a' \in I, |a-a'| \leq k_s |i-i'| + \epsilon_s \AND 
 |i'-a'|\leq e \]
We can satisfy this property by setting $a'=a+i'-i$, $k_s=1$, and
$\epsilon_s=0$.

\paragraph{\Cfour: $\forall a,a' \in A,\forall i \in I, S(i,a) \AND S(i,a') \implies d_A(a,a') \leq \epsilon_t$}
\Cfour \, can be rewritten, once we instantiate $S(i,a)$ to 
\[\exists \epsilon_t, \forall a,a' \in A,\forall i \in I, |i-a|\leq e \AND |i-a'| \leq e \implies |a-a'| \leq \epsilon_t \]
Which is true for $\epsilon_t = 2e$.

%Finally, according to our theorem, the exact semantics of this implementation of CORDIC algorithm is $P_{,,\infty}$.

\section{Example: Dijkstra's shortest path algorithm}\label{sec:dijkstra}

In this section we apply our method   to Dijkstra's shortest path algorithm. This is an algorithm that, given a graph, 
computes the shortest path between a source and any vertex of the
graph. We will prove, by instantiating our schema,  that the following program  implementing the Dijkstra's algorithm can be proved $P_{1,0,0}$
in the semantic of real numbers using our theorem. 

In the following program we use some conventions: the number of
vertices is  fixed to $w$, all 
vertices are connected,  and the maximum  value for a path is $999$ (some
stand-in of infinity).

%\lstset{numbers=left}

{%\scriptsize
\begin{lstlisting}[language=C]
int[] dijkstra( int graph[w][w]){
   int pathestimate[w],mark[w];
   int source,i,j,u,predecessor[w],count=0;
   int minimum(int a[],int m[],int k);
   for(j=1;j<=w;j++){
      mark[j]=0;
      pathestimate[j]=999;
      predecessor[j]=0;}
   source=0;
   pathestimate[source]=0;
   while(count<w){
     u=minimum(pathestimate,mark,w);
     mark[u]=1;
     count=count+1;
     for(i=1;i<=w;i++){
        if(pathestimate[i]>pathestimate[u]+graph[u][i]){
            pathestimate[i]=pathestimate[u]+graph[u][i];
            predecessor[i]=u;}}}
     return pathestimate;}

int minimum(int a[],int m[],int k){
     int mi=999;
     int i,t;
     for(i=1;i<=k;i++){
	  if(m[i]!=1){
	       if(mi>=a[i]){
		    mi=a[i];
		    t=i;}}}
     return t;}

\end{lstlisting}}

\subsection{Scheme instantiation}
To apply our theorem, we have to instantiate the scheme variables  $A$, $B$, $C$ with some variables of the program.  The variables of $I$ are instantiated with the variables that represent the input. 
We choose the following instantiation: $A$ contains the variables $\cnt$ and $\mrk$, $B$ the
array of double $\pe$ and $C$ the variable $u$ which identify the
current vertex to propagate. 

{\scriptsize\begin{lstlisting}
A := int count;int mark[w];
B := pathestimate[w];
C := int u;
I := graph[w][w];
\end{lstlisting}}

We now have to choose a suitable metric on the types of the variables, and we choose the following: $d_I$ is the $L_1$ norm on an
array of real numbers, $d_B$ is the $L_\infty$ norm on array of real numbers  and $d_A$ is the identity metric: that is, the
distance between two elements of $A$ is  $0$ if they are the same
elements and it is $\infty$ otherwise.  

Next, we  identify the stopping condition:
\begin{lstlisting}
S(graph,<count,mark>) := count >= w
\end{lstlisting}
Finally, we identify the functions $M(a,c)$, $N(i,b,c)$, $O(a,b,c,i)$
with the following regions of code:
{%\scriptsize
\begin{lstlisting}[language=C]
O (count, mark, pathestimate, u, graph) {
   u=minimum(pathestimate,mark,w);
   int minimum(int a[],int m[],int k){
	int mi=999;
	int i,t;
	for(i=1;i<=k;i++){
	     if(m[i]!=1){
		  if(mi>=a[i]){
		       mi=a[i];
		       t=i;
		  }
	     }
	}
       return t;
   }
   return u;
}
   
M (<mark, count>, u) {
   mark[u]=1;
   count=count+1;
   return <mark,count>;
}
   
N (graph, pathestimate, u) {
   for(i=1;i<=w;i++){
	if(pathestimate[i]>pathestimate[u]+graph[u][i]){
	     pathestimate[i]=pathestimate[u]+graph[u][i];
	}
   }
   return pathestimate;
}
\end{lstlisting}}

We now have to prove that the conditions \Cone, \Ctwo, \Cthree\, and
\Cfour\, hold for the given instantiations. 

\subsection{Proof of the conditions}

\paragraph{\Cone: $\forall l \in C^*. P_{\Kn,\En,\delta}(\lambda z. \foob(l,z))$}
For all $i0\in I$, $\fooa(i0,i)$ is  $k$-Lipschitz and $k$ does not depend on $i0$.
The proof of this condition can be done by using standard technical (such as Hoare
triples or abstract interpretation) on the following program.
{%\scriptsize
\begin{lstlisting}[language=C]
int[] dijkstra( int graph[w][w], int[] listFoo)
{
    int pathestimate[w],mark[w];
    int source,i,j,u,predecessor[w],count=0;
    int minimum(int a[],int m[],int k);  
    for(j=1;j<=w;j++){
        mark[j]=0;
	pathestimate[j]=999;
	predecessor[j]=0;
    }
    source=0;
    pathestimate[source]=0;
    for(j=0;j<listFoo.length;j++){
        u=listFoo[j];
        for(i=1;i<=w;i++) {
	    if(pathestimate[i]>pathestimate[u]+graph[u][i]){
	         pathestimate[i]=pathestimate[u]+graph[u][i];
		 predecessor[i]=u;
	    }
	}
    }
    return pathestimate;
}

\end{lstlisting}}
In an exact semantics (with real numbers), this program is $1$-Lipschitz as any element of $\pe$ is  the sum of some elements of $\graph$.
If the analysis is done with an exact semantics (with real numbers), we
are able to prove that this program is $1$-Lipschitz.

\paragraph{\Ctwo: $\forall i_1,i \in I, d_I(i,i_1) \leq \delta \implies
d_B( \fooB(i,i) , \fooB(i_1,i)) \leq k_A  d_A( \fooA(i_1) , \fooA(i)) + \epsilon_2$}
The proof for \Ctwo\ is rather technical. The basic idea is however
quite simple. Indeed, the $A$ structure is a set in a discrete space
on which elements are added. So we prove that whatever the order of
the element is $B$ is constant. This is done by showing that local
transpositions do not change the result. So the principle should
apply in other algorithms with the same $A$ structure. 
The complete proof can be found in the appendix of \cite{gazeau12hal}.

\paragraph{\Cthree: $\forall a \in A, \forall i,i' \in I,
d_I(i,i') \leq \delta \AND 
S(i',a) \implies \exists a' \in A, d_A(a,a') \leq k_s d_I(i',i) + \epsilon_s \AND 
S(i,a')$}
Since the instantiation of $S(i',a)$ is \lstinline{count >= w}, 
the stopping condition does not depend on $i$ (when the number of nodes $w$ is fixed).
Hence, the formula is satisfied for $a'= a$ with the constant $k_s=0$ and $\epsilon_s=0$.

\paragraph{\Cfour: $\forall a,a' \in A,\forall i \in I, S(i,a) \AND S(i,a') \implies d_A(a,a') \leq \epsilon_t$}
Since $\{ a | S(i,a)\}$ is a singleton for every $i$ (it corresponds to the state where all the nodes are marked), the property
holds for $\epsilon_t=0$.

\section{Related Work}\label{sec:related}

%Discuss connections to static analysis.

Static analysis via abstract interpretation can be an effective method
for deriving precise bounds on deviations
\cite{goubault01sas,goubault11vmcai}.  Since such static analysis is
generally limited to analyzing code line-by-line, significant over
approximations might be necessary.  For example, when encountering an
``if'' instruction (or a looping construct), a static analyzer will
have to assume that either the control flow is not perturbed by the
finite-precision errors (often unrealistic) or the results from the
two branches of the conditional must be merged (often causing
significant over-approximation).  In our examples here, control flow
can be perturbed a great deal by precision errors and merging both
branches is not a solution as the program is not locally continuous.
Our method is useful for solving this problem since it avoids narrowly
analyzing the semantics of the conditional.

%Cite: paper other paper on Dijkstra's algorithm 

In the two papers \cite{chaudhuri11sigsoft,chaudhuri10popl},
robustness analysis is done for the Dijkstra's algorithm. The authors
split their analysis into two parts: first they prove the continuity
of the algorithm and second they prove it is piecewise robust. The
problem of discontinuity that can occur at some point of the execution
is solved through an abstract language syntax for loops.  Like in our
theorem, this syntax need additional conditions (mainly the
commutativity for two observable equivalent commands).  However, their
abstract language is more specific than our theorem: CORDIC is not in
the scope of these papers which also means their conditions are
simpler and their proofs are more directed than ours.  The other
distinction is in the semantics of the language. Their paper aims at
furnishing the whole semantics which is an exact one and computational
errors are treated qualitatively with the argument that a robust
program is not sensitive to small variations. With our analysis, we
give a quantitative definition of what small enough means.  The last
difference is our design for analyzing non-local-robustness.  We prefer
to consider non-local behaviors as happening and solving them by a program
transformation using pattern than to rewrite the program in a syntax
that hide the non-local behavior.

\section{Future work and conclusion}\label{sec:conclude}
%The properties used are just concatening commands and looking for P and P-

We have presented a theorem that allows us to prove the robustness of some
floating point programs.  This theorem is abstract enough to be
applicable in a number of rather different programs: here, we
illustrate its use with programs to compute cosine using the CORDIC
method and to compute the shortest path in a graph. 

%% DM dropped this since I wasn't clear on its meaning
% This theorem can be used anytime a program contain a mechanism that
% schedules actions and when \emph{if} errors just change the
% scheduling.

%But, more than the efficiency of
%the result it can do, we want to propose another way to conceive
%analysis: a method where we start with the globally of the program to
%go deeper in specific pieces of codes. There are several advantages to
%do this way.  First, the general coordination of piece of code are
%often well known by the programmer. Secondly, it allows once we reach
%a more specific part of the code to have some knowledge about the
%exact property to prove. In particular that means, it helps him to
%choose the appropriate analysis for the actual code. Last point, the
%end of the analysis give macro structure that a programmer may not
%have explicitly formalized, because a function as been used just once
%for instance.

For future work, we would like to address a key possible weakness of
our method: it is currently tied to a particular template.  Although
that template is presented abstractly, there should certainly be ways
to improve the generality beyond the matching of a template.  Also,
since the property $P_{k,\epsilon,\delta}$
(Definition~\ref{def:k-epsilon-delta}) is more general than both
$k$-Lipschitz and the other definitions of robustness
\cite{majumdar09rtss,majumdar10memocode}, we would like to explore
applications of this property to cases where neither of the other
definitions work.% (or give less interesting robustness guarantees).

Condition \Ctwo\ is, at least in the examples considered in this paper,
the most difficult condition to verify.  This suggests that we might
consider more restrictive conditions that would entail \Ctwo.

\paragraph{Acknowledgments:}  We would like to thank Eric Goubault and
Jean Goubault-Larrecq for many useful discussions on the topic of this
paper and for the helpful comments of the anonymous reviewers. 

\bibliographystyle{eptcs}
\bibliography{bib}
\end{document}